\documentclass[a4paper,10pt,reqno, english]{amsart}

\usepackage{amsmath,amssymb,amscd,amsthm,amsfonts}
\usepackage{graphicx,subcaption}
\usepackage{hyperref}
\usepackage{dsfont}
\usepackage{xcolor}
\usepackage{color,soul}
\usepackage[utf8]{inputenc}
\usepackage[nobysame, alphabetic]{amsrefs}
\usepackage[capitalize]{cleveref}
\usepackage{float}

\newtheorem{theorem}{Theorem}[section]
\newtheorem{corollary}[theorem]{Corollary}
\newtheorem{lemma}[theorem]{Lemma}

\theoremstyle{definition}

\newtheorem{definition}[theorem]{Definition}

\newcommand{\rr}{\mathds{R}}

\title{Fair distributions for more participants than allocations}


\author[Sober\'on]{Pablo Sober\'on}\address{Baruch College, City University of New York, New York, NY 10010} 
\email{pablo.soberon-bravo@baruch.cuny.edu}

\thanks{Sober\'on's research is supported by NSF grant DMS 2054419 and a PSC-CUNY TRADB52 award.}

\begin{document}

\begin{abstract}
 We study the existence of fair distributions when we have more guests than pieces to allocate, focusing on envy-free distributions among those who receive a piece.  The conditions on the demand from the guests can be weakened from those of classic cake-cutting and rent-splitting results of Stromquist, Woodall, and Su.  We extend existing variations of the cake-cutting problem with secretive guests and those that resist the removal of any sufficiently small set of guests.
\end{abstract}

\maketitle

\section{Introduction}

Determining the existence of fair or envy-free distributions of goods is a fundamental problem in mathematical economics.  One classic variant is known as the cake-cutting problem.  We have $k$ guests. We aim to divide a cake, represented by the interval $[0,1]$, into $k$ intervals and give one to each guest.  Given a partition, each guest knows which pieces are their favorites.  Cake-cutting problems seek conditions on the guests' subjective preferences to guarantee the existence of a partition and a distribution so that each guest receives one of their favorite pieces.  The appeal of these problems follows from their applications and from their connection to topological combinatorics.

Stromquist and Woodall proved such results with mild conditions \cites{Stromquist1980, Woodall1980}.  Each guest must always prefer a piece of positive length to a piece of length zero (which we call the ``hungry guest'' condition), and the preferences must be closed:  given a sequence of converging partitions, if the $i$-th guest always prefers the $j$-th part, then they must do so in the limit.

Many variations of the cake-cutting problem have been studied \cites{Weller:1985kj, Brams:1996wt, Barbanel2005, procaccia2015cake}, including partitions of multiple cakes or pieces per guest \cites{Cloutier2010, Nyman2020, Aharoni2020, SegalHalevi2021}, partitions with secretive guests \cite{Asada:2018ix}, and partitions weakening the hungry condition if there are pieces of length zero \cites{Meunier2019, Avvakumov2021}.

We study the existence of envy-free divisions if there are more guests than cake pieces.  The number of pieces we cut the cake into, $k$, will be fixed throughout the paper.  We exhibit weak conditions, in which a guest may reject every proposed piece of cake, that still guarantee envy-free distributions.  For positive integers $\alpha \le n$, we say that $n$ guests are \textit{$\alpha$-hungry} if the following two conditions are met:
\begin{itemize}
    \item for any partition of $[0,1]$ into $k$ intervals and any $n-\alpha+1$ of the guests, there is at least one guest who prefers a part of positive length and
    \item for any guest, their set of preferences is closed.
\end{itemize}

\begin{theorem}\label{thm:main}
Let $k \le n$ be positive integers.  We aim to split the interval $[0,1]$ among $k$ out of $n$ guests.  Assume that the $n$ guests are $k$-hungry.  Then, there exists a partition of $[0,1]$ into $k$ intervals that can be distributed among $k$ different guests so that each receives one of their favorite pieces.
\end{theorem}

We interpret rejecting a partition as not listing any piece as a favorite, even pieces of length zero.  The case $n=k$ of the theorem above is the classic cake-cutting result of Stromquist and Woodall.  If among the $n$ guests we can find $k$ hungry guests, we can divide the cake among them.  The result above does not hold with $(k-1)$-hungry guests.  If $n-k+1$ guests never want cake, we are only left with $k-1$ guests who may want a part. This forces us to give one of the parts to a person who does not want it.

The backbone of the proof of \cref{thm:main} is an extension of a classic topological result of Knaster, Kuratowski, and Mazurkiewicz on combinatorial properties of coverings of high-dimensional simplices \cite{Knaster:1929vi}, which we describe in \cref{sec:preliminaries}.  Most variations of cake-cutting results are based on ``colorful'' extensions of such covering theorems, initiated by Gale \cite{Gale1984}.  Our main topological contribution, which directly implies \cref{thm:main}, is \cref{thm:sparse-kkm}.  This is a sparse version of the colorful KKM theorem.

The case $n=k=2$ is often useful to motivate fair partition problems, and can be solved with a ``moving knife'' argument.  The same argument shows a solution for $k=2$ and any $n$.  We first select $n-1$ guests, and make our cut at $x=0$.  At this point all guests who want cake prefer the right side of the cut.  We slide the cut continuously until we reach the first point $x=x_0$ when at least one of the $n-1$ guests prefers the left side of the cut.  By the closed preferences, at least one guest still prefers the right side.  If those are different guests, we have an envy-free distribution.  If the same guest $A$ listed both sides as favorite, we set $A$ apart and ask the remaining $n-1$ guests if anyone wants a piece of cake when we cut at $x_0$.  That guest and $A$ can share the cake.

If one guest's preferences are secret and all guests are hungry, Woodall showed that it is possible to split the cake into $k$ intervals without consulting the secretive guest and still find an envy-free distribution of that partition among the $k$ guests regardless of which piece the secretive guest prefers \cite{Woodall1980}.  An elegant proof was recently found by Asada et al. \cite{Asada:2018ix}.  We extend those results as well, with more guests than pieces of cake.

\begin{theorem}\label{thm:secret-cake}
Let $k \le n$ be positive integers.  Suppose we have $n$ guests who are $k$-hungry.  The preferences of one guest, Alice, are secret.  If we know that Alice is hungry, then it is possible to partition $[0,1]$ into $k$ intervals so that, regardless of which piece Alice prefers, there is an envy-free distribution of the cake among her and $k-1$ other guests.
\end{theorem}

If we drop the assumption that Alice is hungry, the result above still holds but we cannot guarantee that she will get a piece of cake.  However, we can conclude that if she does want one of the pieces in the proposed partition then there is an envy-free distribution where she gets that piece.

Meunier and Su proved other results for cake-cutting with more guests than cake pieces \cite{Meunier2019a}.  One of their results shows that the cake can be divided into $k$ pieces in advance of the $n$ guests arriving, and even if $\lceil (n-k)/k \rceil$ guests miss the party, we can distribute the cake in an envy-free way among $k$ of those who showed up.  Their result also extends to the $\alpha$-hungry setting.

\begin{theorem}\label{thm:people-missing}
Let $k \le r \le n$ be positive integers.  We aim to split the interval $[0,1]$ among $k$ out of $n$ guests.  Assume that the $n$ guests are $r$-hungry.  Then, there exists a partition of $[0,1]$ into $k$ intervals that, regardless of which $\lceil (r-k)/k \rceil$ guests don't show up, we can distribute the cake among $k$ different remaining guests so that each receives one of their favorite pieces.
\end{theorem}

The case $r=n$ is Meunier and Su's result.  We highlight the case $r=k+1$, as it is easier to prove and to interpret.

\begin{corollary}\label{coro:one-fewer}
Let $k < n$ be positive integers.  We aim to split the interval $[0,1]$ among $k$ out of $n$ guests.  Assume that the $n$ guests are $(k+1)$-hungry.  Then, there exists a partition of $[0,1]$ into $k$ intervals so that, regardless of which single person is excluded, the cake can be distributed among $k$ different remaining guests so that each receives one of their favorite pieces.
\end{corollary}

The dual problems to cake-cutting results are known as rent-splitting problems.  In this setting, guests prefer empty pieces over non-empty ones.  We also obtain rent-splitting versions of the theorems mentioned above.  The interpretation, in this case, is that we have $n$ potential tenants and an apartment with $k$ rooms and a fixed rent.  Given a price distribution for the rooms, each potential tenant states which rooms they are willing to rent or if they do not want any room with the current proposal.  The key condition is that if any room is free, they will only accept a free room.

\begin{theorem}\label{thm:rent}
Let $k \le n$ be positive integers.  We aim to rent an apartment with $k$ rooms and have $n$ potential tenants.  Suppose that for any price distribution of the rooms and any $n-k+1$ potential tenants, at least one of them is willing to rent one of the rooms with the proposed price, and their preferences are closed.  Then, we can find a price distribution and $k$ tenants such that each is willing to rent a different room.
\end{theorem}

\cref{thm:rent} is the version of \cref{thm:main} for rent division.  The analog rent division versions of \cref{thm:secret-cake} and \cref{thm:people-missing} also hold although we do not state them explicitly.  The weakening of the hungry condition in \cref{thm:main} is similar to sparse versions of the colorful Carath\'eodory theorem \cites{Holmsen:2016fo, Soberon:2018gn}.  A relaxed version says that \textit{given $n$ sets of points in $\rr^{k-1}$, if the convex hull of the union of any $n-k+1$ sets contains the origin, we can choose one point from each of $k$ of the sets such that the convex hull of the resulting set contains the origin}.

Frick and Zerbib proved a colorful version of a theorem of Komiya regarding coverings of polytopes \cites{Frick2019, Komiya1994} which generalizes both the colorful Carath\'eodory theorem \cite{Barany1982} and the colorful KKM theorem \cite{Gale1984}.  The colorful KKM theorem is the topological backbone behind most cake-cutting results, including those of this manuscript.  Since both the colorful Carath\'eodory theorem and the colorful KKM theorem can be extended as in \cref{thm:main}, it would be interesting to know if the colorful Komiya theorem has such an extension.  After this manuscript was uploaded to a public repository, McGinnis and Zerbib proved such an extension of Komiya's theorem \cite{McGinnis2021}.

Our proofs are topological and rely on the computation of the topological degree of a map, combined with compactness arguments.  We describe the topological and linear-algebraic preliminaries in \cref{sec:preliminaries} and \cref{sec:linear-algebra}, respectively.  The proofs of the theorems described in the introduction follow in \cref{sec:proofs}.




\section{Topological preliminaries}\label{sec:preliminaries}

We can identify a partition of $[0,1]$ into $k$ intervals with a vector $(x_1, \ldots, x_k)$ where $x_j$ is the length of the $j$-th piece, so it is a point in the $(k-1)$-dimensional simplex $\Delta^{k-1}$.  Let $v_1, \ldots, v_k$ be the vertices $\Delta^{k-1}$.  For $j=1,\ldots, k$, let $F_j$ be the facet of $\Delta^{k-1}$ opposite to $v_j$.  We denote by $[n]$ the set $\{1,\ldots, n\}$.

Given a guest, we can denote by $A_j \subset \Delta^{k-1}$ the set of partitions of $[0,1]$ in which they prefer the $j$-th piece.  The conditions described in the introduction imply that $A_j$ must be a closed set.

\begin{definition}
We say that a $k$-tuple $(A_1,\ldots, A_k)$ is a KKM cover of $\Delta^{k-1}$ if for every face $\sigma$ of $\Delta^{k-1}$, we have $\displaystyle \sigma \subset \bigcup_{v_j \in \sigma} A_j$.
\end{definition}

Note that we can take $\sigma = \Delta^{k-1}$, so the union of the sets $A_j$ covers $\Delta^{k-1}$.

\begin{figure}
    \centering
    \includegraphics[width=1\textwidth]{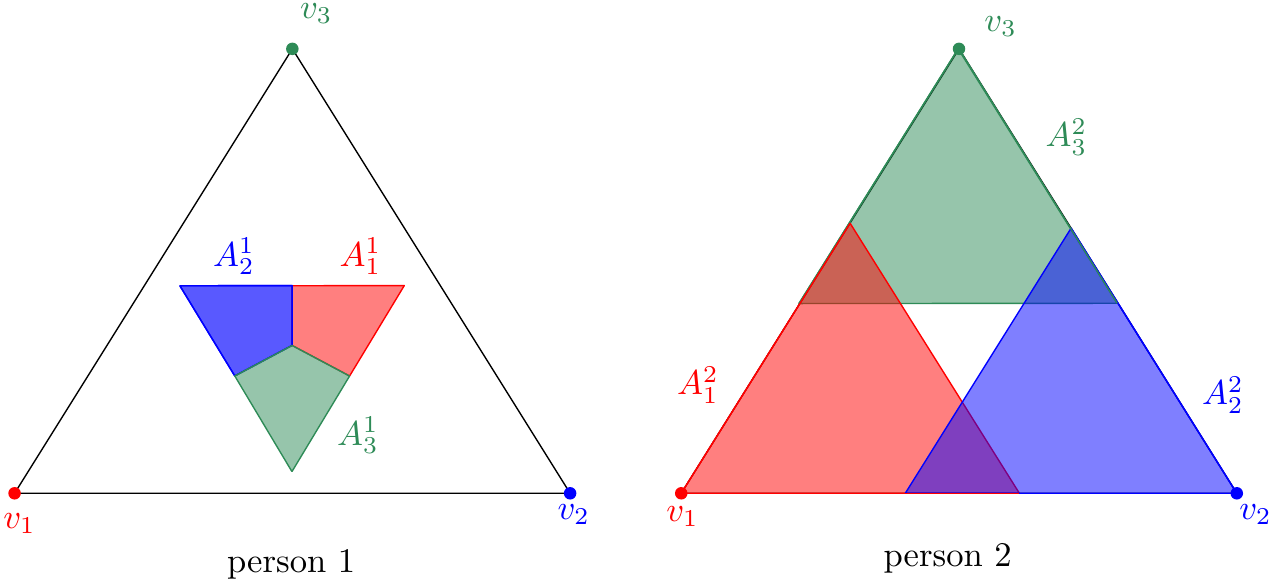}
    \caption{Two preferences for $k=3$.  Person $1$ preferss the smallest piece, but won't have cake if someone has almost half the cake.  Person $2$ prefers a sufficiently large piece of the cake.  Neither is a KKM cover, but their union is a KKM cover.}
    \label{fig:KKM}
\end{figure}

KKM covers of $\Delta^{k-1}$ were defined by Knaster, Kuratowski, and Mazurkiewicz \cite{Knaster:1929vi}.  A KKM cover corresponds to the preferences of a hungry guest.  Knaster, Kuratwoski, and Mazurkiewicz proved that, for any KKM cover, $\bigcap_{j=1}^k A_j \neq \emptyset$.  This translates to the existence of a partition in which the hungry guest does not mind which part they get.  Gale proved the following ``colorful'' version of this result.

\begin{theorem}[Gale 1984 \cite{Gale1984}]
Let $k$ be a positive integer.  For $i=1,\ldots, k$, let $(A^i_1,\ldots, A^i_k)$ be a KKM cover of $\Delta^{k-1}$.  Then, there exists a permutation ${\pi:[k] \to [k]}$ for which $\bigcap_{j=1}^k A^{\pi(j)}_j \neq \emptyset$.
\end{theorem}

In the context of cake-cutting, the result above guarantees the existence of a partition in which guest $\pi(j)$ is content to receive the $j$-th piece, so we have an envy-free distribution.  If we write our condition of $\alpha$-hungry in terms of KKM covers of $\Delta^{k-1}$ we require the following definition.

\begin{definition}
Let $\alpha, k$, and $n$ be positive integers so that $n \ge \alpha$ and $n \ge k$.  For every $i=1,\ldots, n$ we have a $k$-tuple $(A^i_1,\ldots, A^i_k)$ of closed subsets of $\Delta^{k-1}$.  We say that the family of $k$-tuples is \textit{$\alpha$-weakly KKM} if for every $C \in \binom{[n]}{n-\alpha+1}$ the $k$-tuple
\[
\left(\bigcup_{i\in C}A^i_1, \ldots, \bigcup_{i\in C}A^i_k \right)
\]
is a KKM cover of $\Delta^{k-1}$.
\end{definition}

We illustrate an example in \cref{fig:KKM} of two preference triples that could appear in such covers of $\Delta^{2}$.  Finally, for the rent-splitting version of our theorems we use dual KKM covers defined as follows.

\begin{definition}
We say that a $k$-tuple $(A_1,\ldots, A_k)$ is a dual KKM cover of $\Delta^{k-1}$ if
\begin{itemize}
    \item for $j=1,\ldots, k$ the set $A_j$ is a closed subset of $\Delta^{k-1}$,
    \item if $j \neq j'$, then $A_j \cap F_{j'} \subset F_j \cap F_{j'}$, and 
    \item the union of $A_1,\ldots, A_k$ covers $\Delta^{k-1}$.
\end{itemize}
\end{definition}

The condition $A_j \cap F_{j'} \subset F_j \cap F_{j'}$ corresponds to the statement \textit{``if room $j'$ is free, we may only prefer room $j$ over it if it is also free''}.  Alternatively, the condition on the set is equivalent to $A_j \cap \partial \Delta^{k-1}= F_j$.  This is because the relative interior of $F_j$ can only be covered by $A_j$, and since $A_j$ is closed we obtain $F_j \subset A_j$.  For our results, we extend the definition as follows.

\begin{definition}
Let $\alpha, k,$ and $n$ be positive integers so that $n \ge \alpha$ and $n \ge k$.  For every $i=1,\ldots, n$ we have a $k$-tuples $(A^i_1,\ldots, A^i_k)$ of closed subsets of $\Delta^{k-1}$.  We say that the family of $k$-tuples is \textit{$\alpha$-weakly dual KKM} if for every $C \in \binom{n}{n-\alpha+1}$ the $k$-tuple
\[
\left(\bigcup_{i\in C}A^i_1, \ldots, \bigcup_{i\in C}A^i_k \right)
\]
is a dual KKM cover of $\Delta^{k-1}$.
\end{definition}

\section{Linear-algebraic preliminaries}\label{sec:linear-algebra}

In our proofs, we use properties of non-square matrices with non-negative entries with some conditions on their rows and columns.  These are extensions of Birkhoff's theorem on doubly stochastic matrices \cite{Birkhoff1946}.  Birkhoff proved that every doubly stochastic matrix is a convex combination of permutation matrices.  In particular, if $X$ is a $k \times k$ doubly stochastic matrix, there exists a permutation $\pi:[k] \to [k]$ such that its entries $x_{ji}$ satisfy $x_{j\pi(j)}>0$ for all $j \in [k]$.

We extend this result to certain non-square matrices.  The proof relies on bootstrapping Birkhoff's theorem.

\begin{lemma}\label{lem:modified-birkhoff}
Let $k \le n$ be positive integers and $X$ be a $k \times n$ matrix with non-negative entries.  Assume that the sum of rows of $X$ , the sum of each column of $X$ is at most $1$, and there are at least $k$ columns that each sums to $1$.  Then, there exists an injective function $\pi:[k] \to [n]$ such that $x_{j\pi(j)}>0$ for all $j \in [k]$.  Moreover, if $x_{j_0i_0}>0$ for some particular values $j_0 \in [k], i_0 \in [n]$, we may impose the condition $\pi(j_0)=i_0.$
\end{lemma}

\begin{proof}
The case $k=n$ is Birkhoff's theorem, so we assume $k<n$.  Let $s_i$ be the sum of the entries in the $i$-th column of $X$ and $p=\sum_{i=1}^n s_i$.  The sum of each row is $p/k$.  We append $n-k$ rows to $X$ such that each entry added to the $i$-th column is equal to
\[
\frac{1}{n-k}\left(\frac{p}{k}-s_i\right).
\]
First, since $k$ columns add to $1$, we know $p/k \ge 1 \ge s_i$ for all $i$, so each new entry is non-negative.  Second, the sum in each new row is 
\begin{align*}
    \sum_{i=1}^n\left(\frac{1}{n-k}\left(\frac{p}{k}-s_i\right)\right) = \frac{n}{n-k}\cdot \frac{p}{k} - \frac{\sum_{i=1}^n s_i }{n-k} = \frac{n}{n-k}\cdot \frac{p}{k} - \frac{p }{n-k} = \frac{p}{k}.
\end{align*}
Third, the sum of each column is $s_i + (n-k)[(p/k-s_i)/(n-k)] = p/k$.  The new $n \times n$ matrix is a scalar multiple of a doubly stochastic matrix, so we can write it as a linear combination of permutation matrices with only positive coefficients.  Since the entry $(j_0, i_0)$ of $X$ is positive, in the linear combination there must also be a permutation with a $1$ at the $(j_0, i_0)$ entry.  The first $k$ rows of that permutation matrix induce the injective function we were looking for.
\end{proof}

If we know that more columns sum to $1$, we can extend a similar result of Meunier and Su \cite{Meunier2019a}*{Lemma 2.10}.  The proof for $r=k$ gives an alternative proof of \cref{lem:modified-birkhoff} without imposing a single fixed value on $\pi$.

\begin{lemma}\label{lem:modified-su}
Let $k \le r \le n$ be positive integers and $X$ be a $k \times n$ matrix with non-negative entries.  Assume that the sum of rows of $X$ is constant.  The sum of each column of $X$ is at most $1$, and there are at least $r$ columns that each sums to $1$.  Then, for each set $C \in \binom{[n]}{\lceil(r-k)/k \rceil}$ there exists an injective function $\pi:[k] \to [n]\setminus C$ such that $x_{j\pi(j)}>0$ for all $j \in [k]$. 
\end{lemma}

\begin{proof}
The argument of Meunier and Su transfers directly to this setting.  We include it here for the reader's convenience.  We construct a bipartite graph $G$ on $[k]\times [n]$ whose edges have positive weights.  If $x_{ji}>0$ for some $j \in [k], i \in [n]$, we include the edge $(j,i)$ and give it weight $x_{ji}$.  Let $s_i$ be the sum of the $i$-th column of $X$ and $p$ be the sum of all entries of $X$.  The sum of each row must be $p/k\ge r/k$.

We use Hall's marriage theorem to find a matching in $G$ that covers $[k]$.  For a non-empty set $R \subset [k]$, the sum of the weights of edges incident to $R$ is at most the sum of weights of edges incident to its neighborhood $N(R) \subset [n]$.  This implies the following inequalities:
\begin{align*}
    \left(\frac{r}{k}\right)|R| \le \left(\frac{p}{k}\right)|R| & \le \sum_{i \in N(R)} s_i \le |N(R)| \\
    |R| + \left(\frac{r-k}{k}\right) \le \left(1 + \frac{r-k}{k}\right)|R| & \le |N(R)| \\
    |R| + \left\lceil \frac{r-k}{k}\right\rceil & \le |N(R)|.
\end{align*}

Therefore, if we remove any $\lceil (r-k)/k \rceil$ elements from $[n]$, we still have Hall's condition for a matching covering $[k]$.  This matching induces the injective function we were looking for.
\end{proof}

\section{Reduction to facet avoiding KKM covers}

We first rewrite the statement of \cref{thm:main} in terms of $k$-weakly KKM covers.

\begin{theorem}\label{thm:sparse-kkm}
Let $k \le n$ be positive integers.  For $i\in [n]$, let $(A^i_1,\ldots, A^i_k)$ be a $k$-tuple so that they form a $k$-weakly KKM set.  Then, there exists an injective function $\pi:[k] \to [n]$ such that $\bigcap_{j=1}^k A^{\pi(j)}_j \neq \emptyset$.
\end{theorem}

In the next section we will use a slightly stronger condition on the sets, namely that for each $i,j$ we have $A^i_j \in \Delta^{k-1}\setminus F_j$.  In other words, if there are empty pieces, person $i$ never lists those among their favorites.

\begin{lemma}\label{lem:reduction}
It is sufficient to prove \cref{thm:sparse-kkm} for families of sets so that $A^{i}_{j} \subset \Delta^{k-1}\setminus F_j $ for all $i, j$ to know that it holds in general.
\end{lemma}

\begin{proof}
Let $k \le n$ be positive integers and that $(A^i_1, \ldots, A^i_k)$ are $k$-tuples of subsets of $\Delta^{k-1}$ so that they form a $k$-weakly KKM set.

Consider $\Delta^{k-1}$ to be embedded linearly in $\rr^{k-1}$ so that it is a regular simplex centered at the origin. We now construct a set of $k$-tuples $(Y^i_1,\ldots, Y^i_k)$ that form a $k$-weakly cover of $2\Delta^{k-1}$, a scaled copy of $\Delta^{k-1}$ from the origin.  We start by forming some intermediary sets $X^i_j$.  

For $x \in \Delta^{k-1}$, we take $x \in X^i_j$ if and only if $x \in A^{i}_j$.  For $x \in (2\Delta^{k-1})\setminus \Delta^{k-1}$, let $p(x)$ be the closest point of $\Delta^{k-1}$ to $x$.  By convexity and compactness of $\Delta^{k-1}$, the point $p(x)$ is well defined and unique.  Moreover, $p(x) \in \partial \Delta^{k-1}$.  Let $\sigma(x)$ be the smallest face that contains $p(x)$.  We include $x$ in $X^i_j$ if and only if $p(x) \in A^{i}_j$ and $v_j \in \sigma(x)$.  Finally, we let $Y^i_j$ be the closure of $X^i_j$.  See \cref{fig:reduction} for an illustration of the construction.

Since the union of every $n-k+1$ of the $k$-tuples $(A^i_1, \ldots, A^i_k)$ covers $\Delta^{k-1}$, the union of every $n-k+1$ of the $k$-tuples $(Y^i_1,\ldots, Y^i_k)$ covers $2\Delta^{k-1}$.  Note that the $k$-weakly KKM condition is necessary so that every point in $(2\Delta^{k-1})\setminus \Delta^{k-1}$ is covered east one $X^i_j$ from the union.  By the construction of $Y^i_j$ in $(2\Delta^{k-1})\setminus \Delta^{k-1}$, we have that $Y^i_j$ does not intersect the facet of $2\Delta^{k-1}$ opposite to vertex $v_j$.  We also have $Y^i_j \cap \Delta^{k-1} = A^i_j$ for all $i,j$.

Therefore, we can find an injective $\pi: [k] \to [n]$ so that $\bigcap_{j=1}^k Y^{\pi(j)}_j \neq \emptyset$.  Let $x$ be a point in this intersection.  If $x \in \Delta^{k-1}$, we have $x \in \bigcap_{j=1}^k A^{\pi(j)}_j$ and we are done.  If $x \not\in \Delta^{k-1}$, we can extend the ray starting at $p(x)$ in the direction of $x$ until it hits a point $x' \in \partial (2 \Delta^{k-1})$.  The points $x'$ and $x$ are covered by the same sets $Y^i_j$.  However, since $x'$ is in the boundary of $\Delta^{k-1}$, there is at least one value of $j$ so that $x' \not\in Y^i_j$ for all $i\in [n]$, so it was impossible to find an injective function as we claimed.
\end{proof}

\begin{figure}
    \centering
    \includegraphics{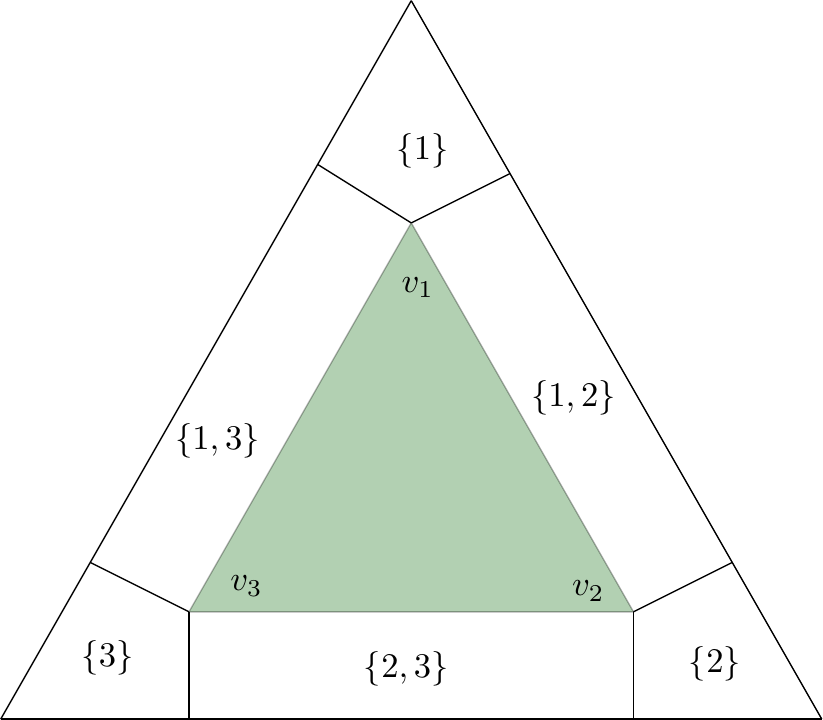}
    \caption{An illustration of the auxiliary subdivision of a simplex used in the construction of the sets $Y^i_j$.  The set $(2\Delta^{k-1})\setminus \Delta^{k-1}$ is divided into regions, each with a set $S$ assigned.  We only allow $x \in Y^i_j$ if $x$ is in the closure of a region for which $j \in S$.  Each facet of $2\Delta^{k-1}$ completely misses one label.}
    \label{fig:reduction}
\end{figure}

\section{main proofs}\label{sec:proofs}

\begin{proof}[Proof of \cref{thm:sparse-kkm}]
By \cref{lem:reduction}, we may assume that $A^i_j \subset \Delta^{k-1}\setminus F_j$ for all $j$.  Let $\tau>0$.  For each $A^i_j$, let $A^i_j(\tau) = \{x \in \Delta^{k-1}: \operatorname{dist}(x,A^i_j)< \tau\}$. This is an open set, and if $\tau$ is sufficiently small we still have $A^i_j(\tau) \subset \Delta^{k-1}\setminus F_j$ for all $i \in [n], j\in [k]$.  Let $B^i_j = \Delta^{k-1}\setminus A^i_j(\tau)$.  Given $\varepsilon>0$ and $x\in \Delta^{k-1}$, we define the $k \times n$ matrix $M(x,\varepsilon)$ with entries $m_{ji}$ given by

\[
m_{ji} = \frac{1}{\max\{\varepsilon, \sum_{h=1}^k \operatorname{dist}(x,B^i_h)\}} \operatorname{dist}(x,B^i_j). 
\]

All the entries are non-negative.  Since the family of $k$-tuples is $k$-weakly KKM, at most $n-k$ columns of $M(x,\varepsilon)$ are zero.

Consider $\Delta^{k-1}$ identified with its natural embedding in $\rr^k$, as the set of vectors with non-negative entries whose sum is $1$.  If we denote by $\bar{u}$ the $n\times 1$ vector with all entries equal to $1$, we can define the function
\begin{align*}
    f_{\varepsilon}: \Delta^{k-1} & \to \Delta^{k-1} \\
    x & \mapsto \frac{1}{\|M(x,\varepsilon) \bar{u}\|_1}M(x,\varepsilon) \bar{u}.
\end{align*}

The function $f_{\varepsilon}$ gives us the sums of rows of $M(x,\varepsilon)$ normalized by the sum of all entries in $M(x,\varepsilon)$.  The function $f_{\varepsilon}$ is continuous.  Moreover, since each $A^i_j(\tau)$ is a subset of $\Delta^{k-1}\setminus F_j$, for every face $\sigma$ of $\Delta^{k-1}$ we have $f_{\varepsilon}(\sigma) \subset \sigma$.  Therefore, $f_{\varepsilon}$ is of degree one on the boundary and must be surjective.  In particular, there exists $x_{\varepsilon} \in \Delta^{k-1}$ such that 
$f_{\varepsilon}(x_{\varepsilon}) = (1/k, \ldots, 1/k)^T$.

Now we take a sequence $(\varepsilon_m)_{m \ge 1}$ of positive real numbers that converges to $0$.  By the compactness of $\Delta^{k-1}$ we may assume without loss of generality that the sequence $(x_{\varepsilon_m})_{m \ge 1}$ converges to a point $x_0^{\tau} \in \Delta^{k-1}$.  For $i \in [n]$, if $x_0^{\tau} \in \bigcup_{j=1}^k A^i_j(\tau)$, then for all sufficiently small values of $\varepsilon$, the $i$-th column of $M(x_0^{\tau},\varepsilon)$ will add to one.  If, on the other hand, $x_0^{\tau} \not\in \bigcup_{j=1}^k A^i_j(\tau)$, then the $i$-th column of $M(x_0^{\tau},\varepsilon)$ will always be zero.  Let $\varepsilon_0$ be a value so that all the columns of $M(x_0^{\tau}, \varepsilon_0)$ sum to zero or one.

We assume without loss of generality that the matrices $M(x_{\varepsilon_m}, \varepsilon_m)$ converge to some matrix $M$.  Note that $M$ may be different from $M(x_0^{\tau}, \varepsilon_0)$.  If $x_0^{\tau} \in \bigcup_{j=1}^k A^i_j(\tau)$, then the $i$-th column of $M$ and of $M(x_0^{\tau}, \varepsilon_0)$ are equal.  Therefore, at least $k$ columns of $M$ sum to one.  Suppose the entry $(j,i)$ of $M$ is positive.  We have $\operatorname{dist} (x_{\varepsilon_m}, B^i_j)>0$ for infinitely many $m$ since otherwise the limit of the $(j,i)$ entry of $M(x_{\varepsilon_m}, \varepsilon_m)$ would converge to zero.  We thus have $x_{\varepsilon_m}\in A^i_j(\tau)$ for infinitely many $m$, which implies $\operatorname{dist}(x_0^{\tau}, A^i_j) \le \tau$.


By the convergence of the matrices $M(x_{\varepsilon_m}, \varepsilon_m)$, we know that $M$ has only non-negative entries, the sum of its rows is constant, and the sum of its columns is bounded above by $1$.  By the coincidence of columns of $M$ and $M(x_0^{\tau}, \varepsilon_0)$, at least $k$ columns of $M$ have sum equal to $1$.

We can therefore apply \cref{lem:modified-birkhoff} and obtain an injective function $\pi^{\tau}:[k]\to [n]$ that selects positive entries of $M$.  In particular, $\operatorname{dist}(x_0^{\tau}, A^{\pi^{\tau}(j)}_j) \le \tau$ for all $j \in [k]$.  As $\tau \to 0$, we may assume without loss of generality that $\pi^{\tau}$ remains constant and is equal to some injective function $\pi:[k] \to [n]$, and $x_0^{\tau}$ converges to some $x_0 \in \Delta^{k-1}$.  Observe that, since the sets $A^i_j$ are closed, $x_0 \in A^{\pi(j)}_j$ for all $j \in [k]$.
\end{proof}

The rest of the proofs are modifications of the proof of \cref{thm:main}.  We start with the secretive version.

\begin{proof}[Proof of \cref{thm:secret-cake}]
By \cref{lem:reduction}, we may assume that $A^i_j \subset \Delta^{k-1}\setminus F_j$ for all $j$ for $i>1$.  We assume that person $1$ is Alice.  Since we don't know her preferences, we provisionally use the $k$-tuple $(B^1_1,\ldots, B^1_k)$ where $B^1_j = F_j$ for each $j \in [k]$ to construct our functions.  We follow the same proof of \cref{thm:main} up until the construction of $M$. 

If the entry $m_{ji}$ of $M$ is not zero, it means $x_0$ is in $A^i_j(\tau)$, so $x_0 \not\in F_j$.  Since none of the rows in $M$ is zero, the point $x_0$ must be in the interior of $\Delta^{k-1}$.  Therefore, the first column has only positive coordinates and coincides with the first column of $M(x_0, \varepsilon_0)$.  Take $j\in[k]$.  We know the entry in row $j$ and column $1$ of $M$ is positive.  We may choose an injective function $\pi_j:[k] \to [n]$ such that $\pi_j(j)=1$ from \cref{lem:modified-birkhoff}.  As $\tau\to 0$ we can assume without loss of generality that each function $\pi_j$ is constant and that $x_0$ converges, finishing the proof.
\end{proof}

A similar approach gives us \cref{coro:one-fewer}.  In the proof of \cref{lem:modified-birkhoff}, if we instead ask that $k+1 \le n$ and $k+1$ columns add to $1$, we need to append at least one row.  Moreover, since $p/k \ge (k+1)/k > 1$, all the entries of this new row are positive.  We can then impose the condition $\pi(k+1) = i$ for any $i \in [n]$ in the function we construct, which is equivalent to excluding person $i$.  \cref{coro:one-fewer} is also a direct corollary of \cref{thm:people-missing}, which we prove now.

\begin{proof}[Proof of \cref{thm:people-missing}]
We follow the same proof of \cref{thm:main} and use \cref{lem:modified-su} instead of \cref{lem:modified-birkhoff}.  As there are only $\binom{n}{\lceil (r-k)/k\rceil}$ possible sets of guests to exclude, the convergence arguments when $\tau \to 0$ still hold.
\end{proof}

\begin{proof}[Proof of \cref{thm:rent}]
The approach is similar to the one of \cref{thm:main}, except we use $k$-weakly dual KKM families of $k$-tuples instead.

First recall that if $(A_1, \ldots, A_k)$ is a dual-KKM cover of $\Delta^{k-1}$, then every point in the relative interior of the facet $F_j$ can only be covered by $A_j$.  Therefore, $A_j \cap \partial \Delta^{k-1} = F_j$.

Given a proper face $\sigma$ of $\Delta^{k-1}$, there must be an index $j \in [n]$ so that $v_j \in \sigma$ and $v_{j+1} \not\in \sigma$, where the sum is taken modulo $n$.  Therefore, $\sigma \subset F_{j+1}$.  This implies that if $(A_1, \ldots, A_k)$ is a dual KKM cover, we can shift the sets by one and obtain $(A_2, \ldots, A_k, A_1)$, which is now a KKM cover of $\Delta^{k-1}$.  The same shift means that the set of $k$-tuples $(A^i_2,\ldots, A^i_k, A^i_1)$ forms a $k$-weakly KKM family.  We apply \cref{thm:sparse-kkm} to this family and we are done.
\end{proof}

\section{Acknowledgments}

The author thanks Francis Su for his valuable comments on this work and the two anonymous referees for their careful revision.


\end{document}